\documentclass[a4paper,11pt,centertags,reqno,psamsfonts]{amsart}

\usepackage{letltxmacro}        
\usepackage{ifthen}		        
\usepackage[headings]{fullpage} 
\usepackage{setspace}		    
\usepackage[inline]{enumitem}   
\usepackage{array}				
\usepackage{fancybox}		    
\usepackage{graphicx}           
\usepackage{subfigure}		    
\usepackage{amsmath}	        
\usepackage{amsthm}				
\usepackage{amssymb}				
\usepackage{mathtools}			
\usepackage{mathrsfs}			
\usepackage{stmaryrd}			
\usepackage{yhmath}				
\usepackage{accents}			
\usepackage{extarrows}          
\usepackage{xfrac}              
\usepackage{url}
\usepackage[all]{xy}			
\usepackage{comment}
\usepackage[round,comma,authoryear,sort&compress]{natbib} 

\usepackage{enumitem}

\newcommand{\tc}{\mathrm{tc}}
\newcommand{\TC}{\mathrm{TC}}

\usepackage[tableposition=bottom]{caption}

\usepackage{soul}
\usepackage{xcolor}

\newcommand{\C}[2]{
\ifthenelse{#1=0 \and #2=0}{\textsf{\upshape C}}
{\ifthenelse{#2=0}{\textsf{\upshape C}^{#1}}
{\textsf{\upshape C}^{#1,#2}}}
}

\makeatletter
\let\oldr@@t\r@@t
\def\r@@t#1#2{%
\setbox0=\hbox{$\oldr@@t#1{#2\,}$}\dimen0=\ht0
\advance\dimen0-0.2\ht0
\setbox2=\hbox{\vrule height\ht0 depth -\dimen0}%
{\box0\lower0.4pt\box2}}
\LetLtxMacro{\oldsqrt}{\sqrt}
\renewcommand*{\sqrt}[2][\ ]{\oldsqrt[#1]{#2}}
\makeatother

\theoremstyle{plain}
\newtheorem{theorem}{Theorem}

\newtheorem{proposition}[theorem]{Proposition}

\theoremstyle{definition}

\theoremstyle{remark}
\newtheorem{remark}[theorem]{Remark}

\numberwithin{equation}{section}
\numberwithin{figure}{section}
\numberwithin{table}{section}

\graphicspath{{}}

\bibpunct{(}{)}{,}{a}{}{,}

\allowdisplaybreaks[4]

\begin{document}
\title[The Impact of Proportional Transaction Costs]{The Impact of Proportional Transaction Costs on Systematically Generated Portfolios}

\author{Johannes Ruf \and Kangjianan Xie}

\address{Johannes Ruf\\
Department of Mathematics\\
London School of Economics and Political Science}

\email{j.ruf@lse.ac.uk}

\address{Kangjianan Xie\\
Department of Mathematics\\
University College London}

\email{kangjianan.xie.14@ucl.ac.uk}

\thanks{We thank Camilo Garc{\'\i}a, Johannes Muhle-Karbe, Soumik Pal, Vassilios Papathanakos, and Leonard Wong for many helpful discussions on the subject matter of this paper.}  



\date{\today}

\begin{abstract} 
The effect of proportional transaction costs on systematically generated portfolios is studied empirically. The performance of several portfolios (the index tracking portfolio, the equally-weighted portfolio, the entropy-weighted portfolio, and the diversity-weighted portfolio) in the presence of dividends and transaction costs is examined under different configurations involving the trading frequency, constituent list size, and renewing frequency. Moreover, a method to smooth transaction costs is proposed.
\end{abstract}

\maketitle

\section{Introduction}

Although often neglected in portfolio analysis for sake of simplicity, transaction costs matter significantly for portfolio performance. Even small proportional transaction costs can have a large negative effect, especially when trades are made to rebalance the portfolio in a relatively high frequency. Hence, one should at least test the performance of a given portfolio when transaction costs are imposed, even if transaction costs are not explicitly taken into account while constructing the portfolio.

In this paper, we examine the effects of imposing transaction costs on systematically generated portfolios, for example, functionally generated portfolios. Such portfolios play a significant role in Stochastic Portfolio Theory; see \cite{MR1894767}. \cite{ruf2018generalised} and \cite{karatzas2018trading} demonstrate empirically that functionally generated portfolios outperform the market portfolio in the absence of transaction costs. To explore whether or to what extent this result still holds when transaction costs are imposed, we empirically examine the performance of portfolios under different configurations including trading frequency, transaction cost rate, constituent list size, and renewing frequency. For the diversity-weighted portfolio, we also propose a method to smooth transaction costs. \cite{wong2019information} indicates an alternative approach, namely to adjust the trading frequency based on certain information-theoretic quantities.

\cite{magill1976portfolio} are among the first to study the impact of proportional transaction costs in portfolio choice. We refer to \cite{MR3183924} and \cite{muhle2017primer} for an overview of the transaction cost literature that evolved afterwards. Most of this literature focuses on the case of one risky asset only. For a discussion of transaction costs in the presence of several risky assets, we refer to \cite{MR2380942}, \cite{MR3032935}, and \cite{MR3418824}. An empirical analysis of the effects of transaction costs is provided in \cite{stoll1983transaction}, \cite{bajgrowicz2012technical}, and \cite{olivares2018robust}. We follow up on this research by providing a systematic analysis of the impact of transaction costs on functionally generated portfolios.

The following is an outline of this paper. Section~\ref{sec data} proposes a framework of backtesting portfolio performance in the presence of transaction costs. In particular, Subsection~\ref{subsec 2.1} incorporates proportional transaction costs when rebalancing a portfolio and Subsection~\ref{subsec 2.2} provides some practical considerations and details when backtesting portfolio performance. Section~\ref{sec E} empirically examines the performance of several different portfolios under various configurations. A method to smooth transaction costs is also provided in Section~\ref{sec E}. Section~\ref{sec C} concludes.

\section{Backtesting in the presence of transaction costs}\label{sec data}

\subsection{Incorporating transaction costs into wealth dynamics}\label{subsec 2.1}

We shall study the performance of long-only stock portfolios that are rebalanced discretely. The market is not assumed to be frictionless; transaction costs are imposed when we trade in the market to rebalance the portfolios. The portfolios are constructed in such a way that their weights match given target weights after paying transaction costs. This construction is more rigid than the one in \cite{garleanu2013dynamic}, for example, where the portfolio weights may deviate from the target weights.

To be more specific, consider a market with $d\geq2$ stocks. Denote the amount of currency invested in each stock by $\psi(\cdot)=(\psi_{1}(\cdot),\cdots,\psi_{d}(\cdot))'$ and the total amount invested in a portfolio by
\[
V(\cdot)=\sum_{i=1}^{d}\psi_{i}(\cdot)\geq0.
\]
Furthermore, denote the portfolio weights by $\pi(\cdot)=(\pi_{1}(\cdot),\cdots,\pi_{d}(\cdot))'$. Note that $\psi_{i}(\cdot)=\pi_{i}(\cdot)V(\cdot)$, for all $i\in\{1,\cdots,d\}$.

Assume that trading stocks involves proportional transaction costs at a time-invariant rate $\tc^{\rm{b}}$ ($\tc^{\rm{s}}$), with $0\leq\tc^{\rm{b}},\tc^{\rm{s}}<1$ for buying (selling) a stock. This means that the sale of one unit of currency of a stock nets only $\left(1-\tc^{\rm{s}}\right)$ units of currency in cash, while buying one unit of currency of a stock costs $\left(1+\tc^{\rm{b}}\right)$ units of currency.

Let us now consider how to trade the stocks in order to match the target weights when transaction costs are imposed. To begin, let us focus on trading at a specific time $t$. When rebalancing the portfolio at time $t$, we know the wealth $\psi(t-)$ invested in each stock and hence the total wealth of the portfolio $V(t-)=\sum_{i=1}^{d}\psi_{i}(t-)$ (exclusive of dividends). We also know the dividends paid at time $t-$, their total denoted by $D(t-)\geq0$.

Given target weights $\pi$, we require $\pi(t)=\pi$ after the portfolio is rebalanced at time $t$. After trading, the wealth $\psi(t)$ invested in each stock in the portfolio satisfies
\begin{equation}\label{eq psi}
\psi_{j}(t)=\pi_{j}(t)\sum_{i=1}^{d}\psi_{i}(t),\quad j\in\{1,\cdots,d\}.
\end{equation}
We provide details about how to compute $\psi(t)$ later in this subsection.

As the portfolio needs to be self-financing, the amount of currency used to buy extra stocks should be exactly the amount of currency obtained from selling redundant stocks plus the dividends if there are any. This yields
\begin{equation}\label{eq Phi1}
\left(1+\tc^{\rm{b}}\right)\sum_{i=1}^{d}\left(\psi_{i}(t)-\psi_{i}(t-)\right)^{+}=\left(1-\tc^{\rm{s}}\right)\sum_{i=1}^{d}\left(\psi_{i}(t-)-\psi_{i}(t)\right)^{+}+D(t-).
\end{equation}
The total transaction costs imposed from trading stocks at time $t$ are computed by
\begin{equation}\label{eq TC}
\TC(t)=\tc^{\rm{b}}\sum_{i=1}^{d}\left(\psi_{i}(t)-\psi_{i}(t-)\right)^{+}+\tc^{\rm{s}}\sum_{i=1}^{d}\left(\psi_{i}(t-)-\psi_{i}(t)\right)^{+}.
\end{equation}
Therefore, the total wealth of the portfolio at time $t$, given by $V(t)=\sum_{i=1}^{d}\psi_{i}(t)$, satisfies
\[
V(t)=V(t-)+D(t-)-\TC(t).
\]

\subsubsection*{Method of computing $\psi(t)$}

In the following, we propose a method to compute $\psi(t)$, given $\psi(t-)$, $D(t-)$, and the target weights $\pi$. Throughout this section, we assume
\[
V(t-)>0,\quad D(t-)\geq0,\quad\sum_{i=1}^{d}\pi_{i}=1,\quad\pi_{j}\geq0,\quad\text{and}\quad\psi_{j}(t-)\geq0,
\]
for all $j\in\{1,\cdots,d\}$.

To begin with, \eqref{eq psi} implies that $\psi(t)$ is of the form
\begin{equation}\label{eq Phi}
\psi_{j}(t)=cV(t-)\pi_{j}(t),\quad j\in\{1,\cdots,d\},
\end{equation}
for some $c>0$. Note that if the market is frictionless, i.e., if $\tc^{\rm{b}}=\tc^{\rm{s}}=0$, and if there are no dividends paid at time $t-$, i.e., if $D(t-)=0$, then $V(t)=V(t-)$ and $c=1$. When transaction costs are imposed, we shall use the constraint \eqref{eq Phi1} to determine $c$.

To make headway, define
\begin{equation}\label{eq Dt}
\widehat{D}=\frac{D(t-)+\left(1-\tc^{\rm{s}}\right)\sum_{i=1}^{d}\psi_{i}(t-)\mathbf{1}_{\pi_{i}(t)=0}}{V(t-)}
\end{equation}
and
\[
c_{j}=\frac{\pi_{j}(t-)}{\pi_{j}(t)}\mathbf{1}_{\pi_{j}(t)>0},\quad j\in\{1,\cdots,d\}.
\]
Then dividing both sides of \eqref{eq Phi1} by $V(t-)$ yields
\begin{equation}\label{eq Psi1}
\left(1+\tc^{\rm{b}}\right)\sum_{i=1}^{d}\left(c-c_{i}\right)^{+}\pi_{i}(t)=\left(1-\tc^{\rm{s}}\right)\sum_{i=1}^{d}\left(c_{i}-c\right)^{+}\pi_{i}(t)+\widehat{D}.
\end{equation}

Note that the LHS of \eqref{eq Psi1} is a continuous function of $c$ and strictly increasing from 0 to $\infty$, as $c$ changes from $\min_{i\in\{1,\cdots,d\}}c_{i}$ to $\infty$. Moreover, the RHS of \eqref{eq Psi1} is a continuous function of $c$ strictly decreasing from $\infty$ to $\widehat{D}\geq0$, as $c$ changes from $-\infty$ to $\max_{i\in\{1,\cdots,d\}}c_{i}$, and equals $\widehat{D}$ afterwards, as $c$ changes from $\max_{i\in\{1,\cdots,d\}}c_{i}$ to $\infty$. Hence, both sides of \eqref{eq Psi1} as functions of $c$ must intersect at some unique point, i.e., a unique solution exists for \eqref{eq Psi1}. To proceed, define
\begin{equation}\label{eq D}
\widehat{D}_{j}=\left(1+\tc^{\rm{b}}\right)\sum_{i=1}^{d}\left(c_{j}-c_{i}\right)^{+}\pi_{i}(t)-\left(1-\tc^{\rm{s}}\right)\sum_{i=1}^{d}\left(c_{i}-c_{j}\right)^{+}\pi_{i}(t),\quad j\in\{1,\cdots,d\}.
\end{equation}
We are now ready to provide an expression for the unknown constant $c$.

\begin{proposition}\label{lemma 2}
Recall that \eqref{eq Dt} and \eqref{eq D} imply $\widehat{D}\geq0$ and $\min_{i\in\{1,\cdots,d\}}\widehat{D}_{i}\leq0$. Hence,
\begin{equation}\label{eq j}
j=\underset{i\in\{1,\cdots,d\}}{\arg\max}\left\{\widehat{D}_{i};\widehat{D}_{i}\leq\widehat{D}\right\}
\end{equation}
is well-defined. Then
\begin{equation}\label{eq c}
c=\frac{\left(1+\tc^{\rm{b}}\right)\sum_{i=1}^{d}c_{i}\pi_{i}(t)\mathbf{1}_{c_{i}\leq c_{j}}+\left(1-\tc^{\rm{s}}\right)\sum_{i=1}^{d}\pi_{i}(t-)\mathbf{1}_{c_{i}>c_{j}}+\widehat{D}}{\left(1+\tc^{\rm{b}}\right)\sum_{i=1}^{d}\pi_{i}(t)\mathbf{1}_{c_{i}\leq c_{j}}+\left(1-\tc^{\rm{s}}\right)\sum_{i=1}^{d}\pi_{i}(t)\mathbf{1}_{c_{i}>c_{j}}}
\end{equation}
solves \eqref{eq Psi1} uniquely.
\end{proposition}

\begin{proof}
By the definition of $\widehat{D}_{j}$ given in \eqref{eq D} and by some basic computations, \eqref{eq c} is equivalent to
\[
c=c_{j}+\frac{\widehat{D}-\widehat{D}_{j}}{\left(1+\tc^{\rm{b}}\right)\sum_{i=1}^{d}\pi_{i}(t)\mathbf{1}_{c_{i}\leq c_{j}}+\left(1-\tc^{\rm{s}}\right)\sum_{i=1}^{d}\pi_{i}(t)\mathbf{1}_{c_{i}>c_{j}}},
\]
which implies $\mathbf{1}_{c_{i}\leq c}\geq\mathbf{1}_{c_{i}\leq c_{j}}$, for all $i\in\{1,\cdots,d\}$.

In the case $\max_{i\in\{1,\cdots,d\}}\widehat{D}_{i}\leq\widehat{D}$, we have $\mathbf{1}_{c_{i}\leq c_{j}}=1$, hence $\mathbf{1}_{c_{i}\leq c}\leq\mathbf{1}_{c_{i}\leq c_{j}}$, for all $i\in\{1,\cdots,d\}$. In the case $\max_{i\in\{1,\cdots,d\}}\widehat{D}_{i}>\widehat{D}$, define
\[
j'=\underset{i\in\{1,\cdots,d\}}{\arg\min}\left\{\widehat{D}_{i};\widehat{D}_{i}>\widehat{D}\right\}.
\]
Then \eqref{eq c} is equivalent to
\[
\begin{aligned}
c&=\frac{\left(1+\tc^{\rm{b}}\right)\sum_{i=1}^{d}c_{i}\pi_{i}(t)\mathbf{1}_{c_{i}<c_{j'}}+\left(1-\tc^{\rm{s}}\right)\sum_{i=1}^{d}\pi_{i}(t-)\mathbf{1}_{c_{i}\geq c_{j'}}+\widehat{D}}{\left(1+\tc^{\rm{b}}\right)\sum_{i=1}^{d}\pi_{i}(t)\mathbf{1}_{c_{i}<c_{j'}}+\left(1-\tc^{\rm{s}}\right)\sum_{i=1}^{d}\pi_{i}(t)\mathbf{1}_{c_{i}\geq c_{j'}}}\\
&=c_{j'}+\frac{\widehat{D}-\widehat{D}_{j'}}{\left(1+\tc^{\rm{b}}\right)\sum_{i=1}^{d}\pi_{i}(t)\mathbf{1}_{c_{i}<c_{j'}}+\left(1-\tc^{\rm{s}}\right)\sum_{i=1}^{d}\pi_{i}(t)\mathbf{1}_{c_{i}\geq c_{j'}}},
\end{aligned}
\]
which implies $\mathbf{1}_{c_{i}>c}\geq\mathbf{1}_{c_{i}>c_{j}}$, for all $i\in\{1,\cdots,d\}$. All in all, we have shown $\mathbf{1}_{c_{i}\leq c}=\mathbf{1}_{c_{i}\leq c_{j}}$, for all $i\in\{1,\cdots,d\}$.

Define next
\begin{gather*}
\Pi^{\rm{b}}=\left(1+\tc^{\rm{b}}\right)\sum_{i=1}^{d}\pi_{i}(t)\mathbf{1}_{c_{i}\leq c_{j}},\quad\Pi^{\rm{s}}=\left(1-\tc^{\rm{s}}\right)\sum_{i=1}^{d}\pi_{i}(t)\mathbf{1}_{c_{i}>c_{j}},\\
\overline{\Pi}^{\rm{b}}=\left(1+\tc^{\rm{b}}\right)\sum_{i=1}^{d}c_{i}\pi_{i}(t)\mathbf{1}_{c_{i}\leq c_{j}},\quad\overline{\Pi}^{\rm{s}}=\left(1-\tc^{\rm{s}}\right)\sum_{i=1}^{d}\pi_{i}(t-)\mathbf{1}_{c_{i}>c_{j}}.
\end{gather*}
Hence, after inserting $c$ by \eqref{eq c} into \eqref{eq Psi1}, the LHS of \eqref{eq Psi1} becomes
\[
\mathrm{LHS}=c\Pi^{\rm{b}}-\overline{\Pi}^{\rm{b}}=\frac{\Pi^{\rm{b}}\overline{\Pi}^{\rm{s}}-\Pi^{\rm{s}}\overline{\Pi}^{\rm{b}}+\Pi^{\rm{b}}\widehat{D}}{\Pi^{\rm{b}}+\Pi^{\rm{s}}},
\]
and the RHS of \eqref{eq Psi1} becomes
\[
\mathrm{RHS}=\overline{\Pi}^{\rm{s}}-c\Pi^{\rm{s}}+\widehat{D}=\frac{\Pi^{\rm{b}}\overline{\Pi}^{\rm{s}}-\Pi^{\rm{s}}\overline{\Pi}^{\rm{b}}-\Pi^{\rm{s}}\widehat{D}}{\Pi^{\rm{b}}+\Pi^{\rm{s}}}+\widehat{D}=\mathrm{LHS}.
\]
Therefore, $c$ defined by \eqref{eq c} indeed solves \eqref{eq Psi1}.
\end{proof}

\begin{remark}
In practice, we can apply both numerical and analytical methods to find the constant $c$. As suggested by \eqref{eq Psi1}, to find $c$ numerically, we can simply search for the minimum of the function
\[
c\mapsto\left|\left(1+\tc^{\rm{b}}\right)\sum_{i=1}^{d}\left(c-c_{i}\right)^{+}\pi_{i}(t)-\left(1-\tc^{\rm{s}}\right)\sum_{i=1}^{d}\left(c_{i}-c\right)^{+}\pi_{i}(t)-\widehat{D}\right|.
\]
Alternatively, by determining the index $j$ given by \eqref{eq j}, we can apply Proposition~\ref{lemma 2} to compute $c$ analytically.

If the analytical approach is implemented, we can speed up the algorithm by making the following observations. We expect the value of $c$ not to be far away from $1$, which is precisely the value in the case of no transaction costs and no dividends. As suggested by the proof of Proposition~\ref{lemma 2}, the family $(\widehat{D}_{i})_{i\in\{1,\cdots,d\}}$ has the same ranking as $(c_{i})_{i\in\{1,\cdots,d\}}$. Therefore, we proceed by ranking all $c_i$'s in ascending order and comparing $\widehat{D}_{k}$ with $\widehat{D}$, where
\[
k=\underset{i\in\{1,\cdots,d\}}{\arg\max}\left\{c_{i};c_{i}\leq1\right\}.
\]
If $\widehat{D}_{k}=\widehat{D}$, then $j=k$ and we are done. If $\widehat{D}_{k}>\widehat{D}$, then we repeatedly compute $\widehat{D}_i$ corresponding to a smaller $c_{i}<c_{k}$ each time until we find the exact index $j$. If $\widehat{D}_{k}<\widehat{D}$, then we simply go the other way around.
\qed
\end{remark}

Proposition~\ref{lemma 2} is applied to determine the constant $c$ used in \eqref{eq Phi} in order to compute $\psi(t)$. Note that, in this subsection, we take $\psi(t-)$ and $D(t-)$ as given. In the next subsection, we discuss how to compute $\psi(t-)$ and $D(t-)$ from the data.

\subsection{Practical considerations}\label{subsec 2.2}

For the preparation of the empirical study in the next section, we now introduce the method used to backtest the portfolio performance.

To begin with, assume that we are given the total market capitalizations and the daily returns for all stocks; denote these processes by $S(\cdot)=(S_{1}(\cdot),\cdots,S_{d}(\cdot))'$ and $r(\cdot)=(r_{1}(\cdot),\cdots,r_{d}(\cdot))'$, respectively. Assume that there are in total $N$ days. For all $l\in\{1,\cdots,N\}$, let $t_{l}$ denote the end of day $l$, at which the end of day total market capitalizations and the daily returns for day $l$ are available. Moreover, if we trade on day $l$, then we call day $l$ a trading day and the trade is made at time $t_{l}$.

Now focus on a specific trading day $l$ with $l\in\{1,\cdots,N\}$ and fix $i\in\{1,\cdots,d\}$ for the moment. In Subsection~\ref{subsec 2.1}, given $\psi(t_{l}-)$ and $D(t_{l}-)$, as well as the target weights specified by the corresponding portfolio at time $t_{l}$, we have shown how to compute $\psi(t_{l})$. In the following, we show how to obtain $\psi(t_{l}-)$ and $D(t_{l}-)$.

The daily return $r_{i}(t_{l})$ includes the dividends of stock $i$ if there are any. We decompose the daily return $r_{i}(t_{l})$ into two parts: the dividend rate $r^{D}_{i}(t_{l})$ and the realised rate $r^{R}_{i}(t_{l})$. The dividend rate $r^{D}_{i}(t_{l})$ is computed as
\begin{equation}\label{eq rD}
r^{D}_{i}(t_{l})=\max\left\{1+r_{i}(t_{l})-\frac{S_{i}(t_{l})}{S_{i}(t_{l-1})},0\right\}
\end{equation}
and yields the amount of dividends received at time $t_{l}$ for each unit of currency invested in stock $i$ at time $t_{l-1}$\footnote{The dividends computed from the dividend rate $r^{D}$ contain not only the actual stock dividends, but also other corporate actions. For example, AT\&T, which dominated the telephone market for most of the 20$^{\mathrm{th}}$ century, was broken up into eight smaller companies in 1984. This lead to a significant drop in the stock price. In our analysis below, we assume that the investor obtained cash in exchange (instead of stocks in the newly established companies).}. The realised rate $r^{R}_{i}(t_{l})$ is computed as
\[
r^{R}_{i}(t_{l})=r_{i}(t_{l})-r^{D}_{i}(t_{l})
\]
and yields the units of currency held in stock $i$ at time $t_{l}$ for each unit of currency invested in stock $i$ at time $t_{l-1}$.

The maximum is used in \eqref{eq rD} to make sure that the dividend rate is nonnegative. Indeed, occasionally the data may suggest $S_{i}(t_{l-1})(1+r_{i}(t_{l}))<S_{i}(t_{l})$. This can happen, for example, when company $i$ issues extra stocks at time $t_{l}$. In this case, we simply assume that there are no dividends paid at time $t_{l}$. 

A special situation requires us to pay extra attention. A few times, some stock $i$ is delisted from the market at time $t_{l}$, for example, due to bankruptcy or merger. In this case, we still have data for $r_{i}(t_{l})$, but not for $S_{i}(t_{l})$. To deal with this situation, we assume that there are no dividends paid in stock $i$ at time $t_{l}$. As a result, we have $r^{D}_{i}(t_{l})=0$ and $r^{R}_{i}(t_{l})=r_{i}(t_{l})$ for such stock $i$. To close the position in stock $i$, we assume that one needs to pay transaction costs.

Without loss of generality, assume that there are $n\geq1$ days (including the trading day $l$) involved since the last trading day, i.e., the last trading day before $l$ is $l-n$. For all $k\in\{l-n+1,\cdots,l\}$, we compute $r^{D}(t_{k})$ and $r^{R}(t_{k})$ as above. In particular, if some stock $i$ in the portfolio is delisted from the market at time $t_{u}$, for some $u\in\{l-n+1,\cdots,l-1\}$, then we set $r^{R}_{i}(t_{v})=r^{D}_{i}(t_{v})=0$, for all $v\in\{u+1,\cdots,l\}$.

Then given $\psi(t_{l-n})$, we compute
\[
\psi_{i}(t_{l}-)=\psi_{i}(t_{l-n})\prod_{k=l-n+1}^{l}\left(1+r^{R}_{i}(t_{k})\right),\quad i\in\{1,\cdots,d\}.
\]
Since all dividends paid between two consecutive trading days are only reinvested at time $t_{l}$, the total dividends available for reinvesting are computed by
\[
D(t_{l}-)=\sum_{i=1}^{d}\psi_{i}(t_{l-n})\sum_{k=l-n+1}^{l}r^{D}_{i}(t_{k})\prod_{u=l-n+1}^{k-1}\left(1+r^{R}_{i}(t_{u})\right).
\]

\section{Examples and empirical results}\label{sec E}

In this section, we analyze the performance of several portfolios empirically. The target weights are expressed in terms of the market weights $\mu(\cdot)=\big(\mu_{1}(\cdot),\cdots,\mu_{d}(\cdot)\big)'$ with components
\[
\mu_{j}(\cdot)=\frac{S_{j}(\cdot)}{\sum_{i=1}^{d}S_{i}(\cdot)},\quad j\in\{1,\cdots,d\}.
\]
In Subsection~\ref{sec div}, we also propose a method to smooth transaction costs.

We shall consider the largest $d$ stocks. We will vary the number $d$ between 100, 300, and 500. The constituent list (the list of the top $d$ stocks) is renewed either weekly, monthly, or quarterly. Whenever we renew the constituent list, we keep the $d$ stocks with the largest total market capitalizations at that time. We trade only these $d$ stocks afterwards until we renew the constituent list again. If any of these stocks stops to exist in the market due to any reason, we simply invest in the remaining stocks without adding a new stock to the list before we renew it next time. Note that renewing the constituent list implies trading to replace the old top $d$ stocks with the new top $d$ stocks. We trade with a specific frequency, which can be either daily, weekly, or monthly. For research on optimal trading frequency, we refer to \cite{MR3766056}.

At time $t_{0}$, we take the transaction costs due to initializing a portfolio as sunk cost, i.e., we set $\TC(t_0)=0$. Moreover, we start a portfolio with initial wealth $V(t_0)=1000$. Note that unless otherwise mentioned, the logarithmic scale is used when plotting $V(\cdot)$ and $\TC(\cdot)$ for the purpose of better interpretability. To simplify the analysis, we impose a uniform transaction cost rate $\tc$ on both buying and selling the stocks, i.e., we set $\tc^{\rm{b}}=\tc^{\rm{s}}=\tc$.

For each example, we provide tables with the yearly returns, the excess returns (relative to the corresponding index tracking portfolio), the standard deviations of the yearly returns, the Sharpe ratios\footnote{To compute the Sharpe ratios of the portfolios and the indices, the one-year U.S.~Treasury yields are used. The data of these yields can be downloaded from \url{https://www.federalreserve.gov}.}, and the wealth and the cumulative transaction costs at the end of the investment period of the portfolios.

\subsection*{Data source}

The data of the total market capitalizations $S(\cdot)$ and the daily returns $r(\cdot)$ is downloaded from the CRSP US Stock Database\footnote{See \url{http://www.crsp.com/products/research-products/crsp-us-stock-databases} for details.}. This database contains the traded stocks on all major US exchanges. More precisely, we focus on ordinary common stocks\footnote{Those stocks in CRSP which have `Share Code' 10, 11, or 12.}. The data starts January 2$^{\rm{nd}}$, 1962 and ends December 30$^{\rm{th}}$, 2016.

The total market capitalizations are computed by multiplying the numbers of outstanding shares with the share prices, and are essential in determining the target weights. The daily returns include dividends but also delisting returns in case stocks get delisted (for example, the recovery rate in case a traded firm goes bankrupt).

\subsection{Index tracking portfolio}\label{sec I}

In this subsection, we introduce the index tracking portfolio. This portfolio is used to benchmark the performance of other portfolios studied in the following subsections. The index tracking portfolio has target weights
\[
\pi_{j}(\cdot)=\mu_{j}(\cdot),\quad j\in\{1,\cdots,d\}.
\]
Note that this portfolio is rebalanced only when the constituent list changes or when dividends are reinvested. 

The index tracking portfolio includes the effects of paying transaction costs and reinvesting dividends. In contrast, the capitalization index with wealth process
\[
\sum_{i=1}^{d}S_{i}(\cdot)\times\frac{1000}{\sum_{i=1}^{d}S_{i}(t_0)}
\]
does not take transaction costs and dividends into consideration.

In the following, we examine the performance of the index tracking portfolio under different trading frequencies, renewing frequencies, as well as constituent list sizes $d$, when there are no transaction costs, i.e., when $\tc=0$, and when $\tc=0.5\%$ and $\tc=1\%$, respectively. These numbers are consistent with the transaction cost estimates in \cite{stoll1983transaction}, \cite{keim1997transactions}, and \cite{fong2017best}.

\subsubsection*{Varying the trading frequency}

We fix the constituent list size $d=100$ and use monthly renewing frequency.  Table~\ref{tab 8} shows the performance of the index tracking portfolio and the corresponding capitalization index under daily, weekly, and monthly trading frequencies, respectively. Note that the capitalization index does not depend on the trading frequency. As expected, with the same trading frequency, the portfolio performs worse under a larger transaction cost rate $\tc$. In addition, the portfolio outperforms the corresponding index, which implies that the dividends paid exceed the transaction costs imposed even if $\tc=1\%$. In Figure~\ref{fg CI_1}, the wealth processes of the daily traded index tracking portfolio and the corresponding capitalization index are plotted.

\begin{table}[h!]
\begin{tabular}{ l | c | c c c | c c c | c c c }
 & CI & IT$^{d}_{0}$ & IT$^{d}_{0.5}$ & IT$^{d}_{1}$ & IT$^{w}_{0}$ & IT$^{w}_{0.5}$ & IT$^{w}_{1}$ & IT$^{m}_{0}$ & IT$^{m}_{0.5}$ & IT$^{m}_{1}$ \\
 \hline
Yearly return & 8.84 & 10.30 & 10.09 & 9.89 & 10.30 & 10.10 & 9.90 & 10.27 & 10.08 & 9.89 \\
 \hline
Std & 16.59 & 16.87 & 16.84 & 16.81 & 16.88 & 16.85 & 16.82 & 16.88 & 16.86 & 16.83 \\
 \hline
Sharpe ratio & 0.22 & 0.30 & 0.29 & 0.28 & 0.30 & 0.29 & 0.28 & 0.30 & 0.29 & 0.28 \\
\hline
Wealth & 54.5 & 111.7 & 100.7 & 90.7 & 111.3 & 100.7 & 91.0 & 109.7 & 99.7 & 90.6\\
\hline
TC & & & 2.7 & 5.0 & & 2.5 & 4.7 & & 2.3 & 4.3\\
\hline
\end{tabular}
\caption{Yearly returns in percentage, standard deviations of yearly returns (Std), Sharpe ratios, and the wealth and the cumulative transaction costs (TC) in thousands at the end of the investment period of the index tracking portfolio (IT) and the corresponding capitalization index (CI) under different trading frequencies and transaction cost rates $\tc$ with $d=100$ and monthly renewing frequency. The subscript $x$ corresponds to $\tc=x\%$ and the superscripts $d$, $w$, and $m$ indicate daily, weekly, and monthly trading frequencies, respectively.}
\label{tab 8}
\end{table}

\begin{figure}[h!]
\includegraphics[width=\textwidth]{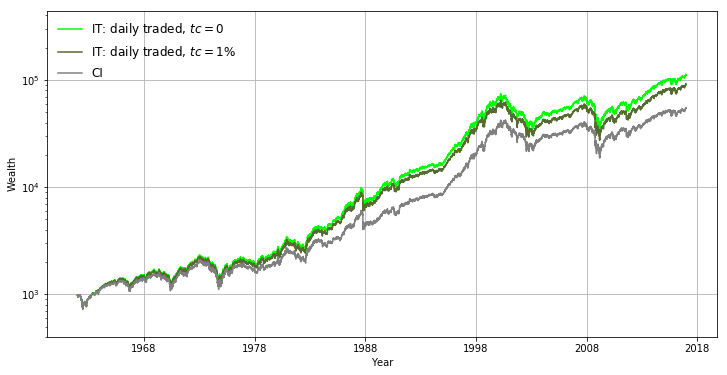}
\caption{The wealth processes of the index tracking portfolio (IT) and the corresponding capitalization index (CI) on logarithmic scale under different transaction cost rates $\tc$ with $d=100$, daily trading frequency, and monthly renewing frequency. The weekly and the monthly traded portfolio performs similarly to the daily traded portfolio under the same transaction cost rate $\tc$.}
\label{fg CI_1}
\end{figure}

\subsubsection*{Varying the renewing frequency}

Still fixing the constituent list size $d=100$, we now use daily trading frequency and vary the renewing frequency between weekly, monthly, and quarterly frequencies, respectively. As shown in Figure~\ref{fg CI_2} and Table~\ref{tab 9}, under the same transaction cost rate $\tc$, the less frequently the constituent list is renewed, the better the portfolio performs. As trades are made when we renew the constituent list, renewing more frequently will impose larger transaction costs, which impacts the performance of the portfolio to a higher degree. Additionally, the more frequently the constituent list is renewed, the more sensitive the portfolio is to a larger transaction cost rate $\tc$.

\begin{figure}
\includegraphics[width=\textwidth]{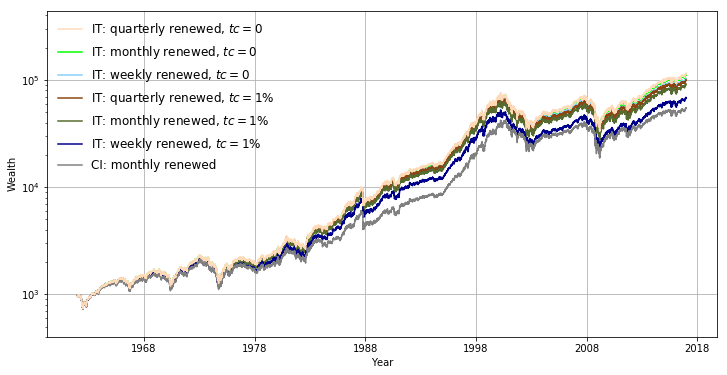}
\caption{The wealth processes of the index tracking portfolio (IT) and the corresponding capitalization index (CI) on logarithmic scale under different renewing frequencies and transaction cost rates $\tc$ with $d=100$ and daily trading frequency. The performance of the weekly renewed portfolio when $\tc=0$ is similar to that of the quarterly renewed portfolio when $\tc=1\%$. The weekly and the quarterly renewed capitalisation index is not very different from the monthly renewed one.}
\label{fg CI_2}
\end{figure}

\begin{table}
\begin{tabular}{ l | c c c c | c c c c }
 & CI$^{W}$ & IT$^{W}_{0}$ & IT$^{W}_{0.5}$ & IT$^{W}_{1}$ & CI$^{Q}$ & IT$^{Q}_{0}$ & IT$^{Q}_{0.5}$ & IT$^{Q}_{1}$ \\
 \hline
Yearly return & 8.85 & 10.14 & 9.73 & 9.33 & 8.82 & 10.34 & 10.20 & 10.06 \\
 \hline
Std & 16.66 & 16.89 & 16.84 & 16.79 & 16.44 & 16.83 & 16.81 & 16.79 \\
 \hline
Sharpe ratio & 0.22 & 0.29 & 0.27 & 0.25 & 0.22 & 0.31 & 0.30 & 0.29 \\
\hline
Wealth & 54.5 & 102.2 & 83.5 & 68.1 & 54.4 & 114.2 & 106.5 & 99.2 \\
\hline
TC & & & 4.2 & 7.2 & & & 2.0 & 3.8 \\
\hline
\end{tabular}
\caption{Yearly returns in percentage, standard deviations of yearly returns (Std), Sharpe ratios, and the wealth and the cumulative transaction costs (TC) in thousands at the end of the investment period of the index tracking portfolio (IT) and the corresponding capitalization index (CI) under different renewing frequencies and transaction cost rates $\tc$ with $d=100$ and daily trading frequency. The subscript $x$ corresponds to $\tc=x\%$ and the superscripts $W$ and $Q$ indicate weekly and quarterly renewing frequencies, respectively.}
\label{tab 9}
\end{table}

\subsubsection*{Varying the constituent list size $d$}

With daily trading and monthly renewing frequencies, we now backtest the performance of the index tracking portfolio under different constituent list sizes $d$. As shown in Figure~\ref{fg CI_3} and Table~\ref{tab 10}, the portfolio outperforms the corresponding index even with transaction cost rate $\tc=1\%$. The more stocks the constituent list contains, the better the portfolio performs.

\begin{figure}
\includegraphics[width=\textwidth]{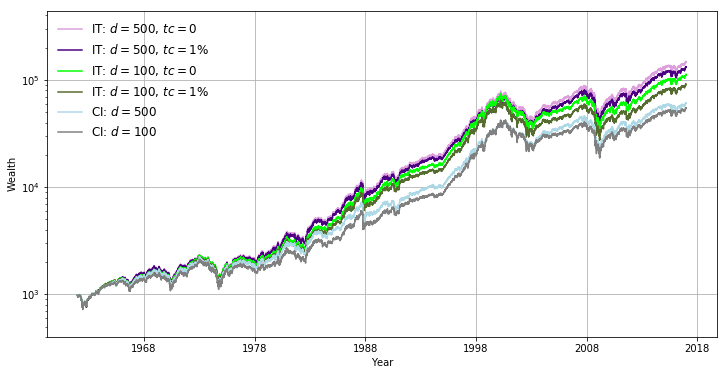}
\caption{The wealth processes of the index tracking portfolio (IT) and the corresponding capitalization index (CI) on logarithmic scale under different constituent list sizes $d$ and transaction cost rates $\tc$ with daily trading and monthly renewing frequencies. For both the portfolio and the index, the wealth processes with $d=300$ are omitted. Everything else equal, they would lie between the plotted ones with $d=100$ and with $d=500$.}
\label{fg CI_3}
\end{figure}

\begin{table}
\begin{tabular}{ l | c c c c | c c c c }
 & CI$^{300}$ & IT$^{300}_{0}$ & IT$^{300}_{0.5}$ & IT$^{300}_{1}$ & CI$^{500}$ & IT$^{500}_{0}$ & IT$^{500}_{0.5}$ & IT$^{500}_{1}$ \\
 \hline
Yearly return & 8.94 & 10.61 & 10.46 & 10.31 & 9.01 & 10.83 & 10.71 & 10.59 \\
 \hline
Std & 16.14 & 16.57 & 16.55 & 16.53 & 16.15 & 16.61 & 16.60 & 16.58 \\
 \hline
Sharpe ratio & 0.23 & 0.33 & 0.32 & 0.31 & 0.24 & 0.34 & 0.33 & 0.33 \\
\hline
Wealth & 58.7 & 132.5 & 123.1 & 114.3 & 60.6 & 147.2 & 139.0 & 131.1 \\
\hline
TC & & & 2.4 & 4.5 & & & 2.3 & 4.3 \\
\hline
\end{tabular}
\caption{Yearly returns in percentage, standard deviations of yearly returns (Std), Sharpe ratios, and the wealth and the cumulative transaction costs (TC) in thousands at the end of the investment period of the index tracking portfolio (IT) and the corresponding capitalization index (CI) under different constituent list sizes $d$ and transaction cost rates $\tc$ with daily trading and monthly renewing frequencies. The subscript $x$ corresponds to $\tc=x\%$ and the superscripts $300$ and $500$ indicate $d=300$ and $d=500$, respectively.}
\label{tab 10}
\end{table}

\subsection{Equally-weighted portfolio}

This subsection examines the equally-weighted portfolio (see \cite{benartzi2001naive} and \cite{windcliff20041} for a discussion of this portfolio in the context of defined contribution plans, and \cite{demiguel2007optimal} for a deep study of its properties). Here, the target weights are given by
\[
\pi_{j}(\cdot)=\frac{1}{d},\quad j\in\{1,\cdots,d\}.
\]

For each portfolio with a specific trading frequency, a specific  renewing frequency, and a specific constituent list size $d$, we examine its performance when there are no transaction costs, i.e., when $\tc=0$, and when $\tc=0.5\%$ and $\tc=1\%$, respectively. As shown in the following, the equally-weighted portfolio outperforms the corresponding index tracking portfolio when there are no transaction costs. This well-behaved performance of the equally-weighted portfolio within a frictionless market is popular in the academic literature. However, the equally-weighted portfolio is very sensitive to transaction costs. Its performance is strongly compromised even with a small transaction cost rate $\tc=0.5\%$.

\subsubsection*{Varying the trading frequency}

Let us fix $d=100$ and apply monthly renewing frequency. Figure~\ref{fg EW_1} plots and Table~\ref{tab 6} summarises the wealth processes of the equally-weighted and the corresponding index tracking portfolio under different trading frequencies and transaction cost rates $\tc$. When there are no transaction costs, i.e., when $\tc=0$, the equally-weighted portfolio outperforms the corresponding index tracking portfolio under all three different trading frequencies. A similar observation is also provided in \cite{banner2018diversification}. In addition, the more frequently the portfolio is traded, the better it performs. Trading more frequently also allows to reinvest the dividends faster, which helps to enhance the portfolio performance.

\begin{figure}
\includegraphics[width=\textwidth]{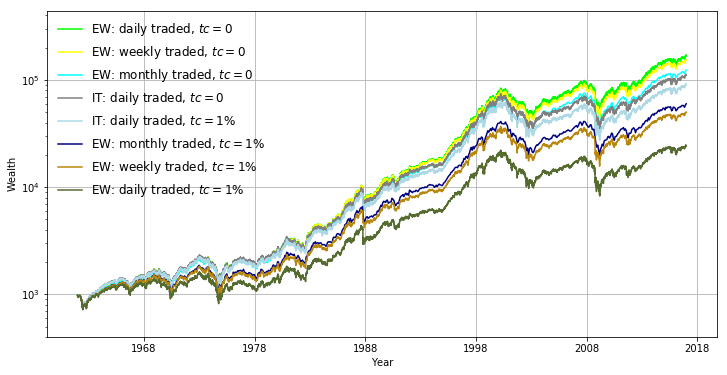}
\caption{The wealth processes of the equally-weighted portfolio (EW) and the corresponding index tracking portfolio (IT) on logarithmic scale under different trading frequencies and transaction cost rates $\tc$ with $d=100$ and monthly renewing frequency. Under the same transaction cost rate $\tc$, the weekly and the monthly traded index tracking portfolio performs similarly to the one traded daily.}
\label{fg EW_1}
\end{figure}

\begin{table}
\begin{tabular}{ l | c c c | c c c | c c c }
 & EW$^{d}_{0}$ & EW$^{d}_{0.5}$ & EW$^{d}_{1}$ & EW$^{w}_{0}$ & EW$^{w}_{0.5}$ & EW$^{w}_{1}$ & EW$^{m}_{0}$ & EW$^{m}_{0.5}$ & EW$^{m}_{1}$ \\
 \hline
Yearly return & 11.10 & 9.19 & 7.31 & 10.94 & 9.82 & 8.72 & 10.53 & 9.81 & 9.10 \\
 \hline
Excess return & 0.80 & -0.9 & -2.58 & 0.64 & -0.28 & -1.18 & 0.26 & -0.27 & -0.79 \\
 \hline
Std & 16.83 & 16.65 & 16.48 & 16.93 & 16.81 & 16.69 & 17.00 & 16.91 & 16.83 \\
 \hline
Sharpe ratio & 0.35 & 0.24 & 0.13 & 0.34 & 0.28 & 0.21 & 0.31 & 0.27 & 0.23 \\
\hline
Wealth & 168.2 & 64.0 & 24.3 & 153.8 & 87.7 & 50.0 & 123.7 & 86.1 & 59.9 \\
\hline
TC & & 15.5 & 15.0 & & 11.4 & 14.8 & & 7.2 & 10.9 \\
\hline
\end{tabular}
\caption{Yearly returns and excess returns (with respect to the index tracking portfolio shown in Table~\ref{tab 8}) in percentage, standard deviations of yearly returns (Std), Sharpe ratios, and the wealth and the cumulative transaction costs (TC) in thousands at the end of the investment period of the equally-weighted portfolio (EW) under different trading frequencies and transaction cost rates $\tc$ with $d=100$ and monthly renewing frequency. The subscript $x$ corresponds to $\tc=x\%$ and the superscripts $d$, $w$, and $m$ indicate daily, weekly, and monthly trading frequencies, respectively.}
\label{tab 6}
\end{table}

When transaction costs are imposed, Figure~\ref{fg EW_1} and Table~\ref{tab 6} suggest that under the same transaction cost rate $\tc$, the more frequently the portfolio is traded, the larger the decrease in portfolio performance is. The performance of the equally-weighted portfolio is strongly affected by transaction costs. Even with $\tc=0.5\%$, the corresponding index tracking portfolio outperforms the equally-weighted portfolio. However, slowing down trading helps to reduce the influence of transaction costs. Indeed, the performance of the monthly traded equally-weighted portfolio when $\tc=1\%$ is similar to that of the daily traded one when $\tc=0.5\%$. As shown in Figure~\ref{fg EW_2}, the cumulative transaction costs paid from a monthly traded equally-weighted portfolio when $\tc=1\%$ are smaller than that from a daily traded one when $\tc=0.5\%$.

\begin{figure}
\includegraphics[width=\textwidth]{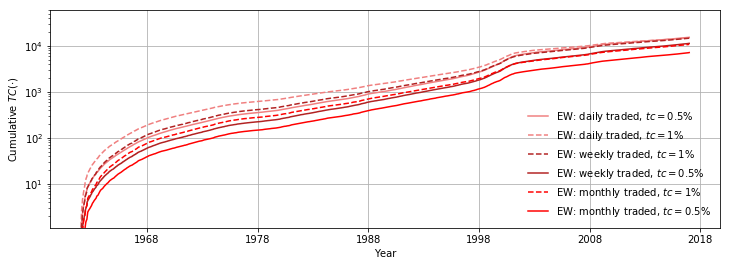}
\caption{Cumulative transaction costs on logarithmic scale of the equally-weighted portfolio (EW) under different trading frequencies and transaction cost rates $\tc$ with $d=100$ and monthly renewing frequency.}
\label{fg EW_2}
\end{figure}

We now study the sensitivity of the Sharpe ratio with respect to the transaction cost rate $\tc$. Specifically, we compute the Sharpe ratios of the monthly traded equally-weighted and index tracking portfolio for $\tc\in\{0,0.01\%,0.02\%,\cdots,0.5\%\}$. As plotted in Figure~\ref{fg EW_3}, the Sharpe ratios of both the equally-weighted and the index tracking portfolio decrease as $\tc$ becomes larger. On the left hand side of the intersection when $\tc<0.22\%$, the equally-weighted portfolio has a higher Sharpe ratio. On the right hand side of the intersection when $\tc>0.22\%$, the inverse situation holds. This indicates that the equally-weighted portfolio depends more on transaction costs than the index tracking portfolio.

\begin{figure}
\includegraphics[width=\textwidth]{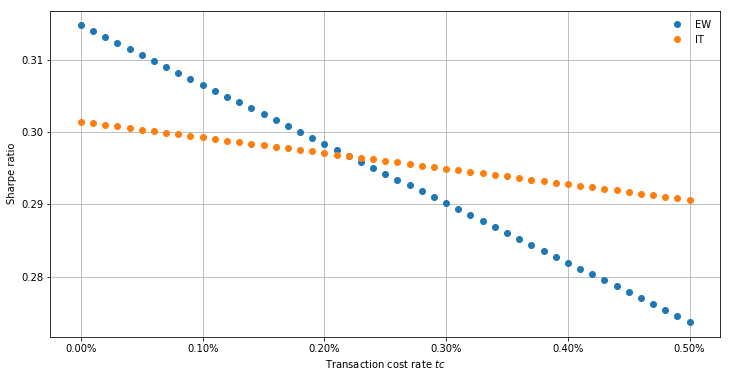}
\caption{Sharpe ratios of the equally-weighted portfolio (EW) and the index tracking portfolio (IT) under different transaction cost rates $\tc$ with $d=100$, monthly trading frequency, and monthly renewing frequency.}
\label{fg EW_3}
\end{figure}

\subsubsection*{Varying the renewing frequency}

Now we examine the performance of the equally-weighted portfolio with $d=100$, daily trading frequency, and under weekly, monthly, and quarterly renewing frequencies, respectively. As shown in Figure~\ref{fg EW_4} and Table~\ref{tab 11}, under the same transaction cost rate $\tc$, the less frequently the constituent list is renewed, the better the portfolio performs. With $\tc=0.5\%$, the equally-weighted portfolio already performs worse than the corresponding index tracking portfolio. In particular, the portfolio with a more frequent renewing frequency is more sensitive to transaction costs. As studied in more detail in Subsection~\ref{sec div}, the reason behind these observations is that trading on renewing days incurs extremely large transaction costs compared with trading on other days when the constituent list is not renewed. These large transaction costs paid on renewing days strongly impact the portfolio performance.

\begin{figure}
\includegraphics[width=\textwidth]{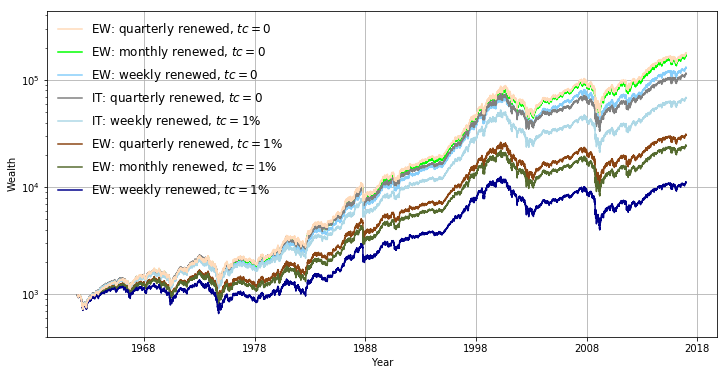}
\caption{The wealth processes of the equally-weighted portfolio (EW) and the corresponding index tracking portfolio (IT) on logarithmic scale under different renewing frequencies and transaction cost rates $\tc$ with $d=100$ and daily trading frequency. For the index tracking portfolio, the wealth processes of the quarterly renewed one with $\tc=0$ and the weekly renewed one with $\tc=1\%$ are plotted. The omitted wealth processes of the index tracking portfolio lie between the plotted ones.}
\label{fg EW_4}
\end{figure}

\begin{table}
\begin{tabular}{ l | c c c | c c c }
 & EW$^{W}_{0}$ & EW$^{W}_{0.5}$ & EW$^{W}_{1}$ & EW$^{Q}_{0}$ & EW$^{Q}_{0.5}$ & EW$^{Q}_{1}$ \\
 \hline
Yearly return & 10.62 & 8.20 & 5.83 & 11.21 & 9.47 & 7.76 \\
 \hline
Excess return & 0.48 & -1.53 & -3.50 & 0.87 & -0.73 & -2.30 \\
 \hline
Std & 16.95 & 16.71 & 16.50 & 16.82 & 16.65 & 16.50 \\
 \hline
Sharpe ratio & 0.32 & 0.18 & 0.04 & 0.36 & 0.26 & 0.16 \\
\hline
Wealth & 129.8 & 37.8 & 11.0 & 177.6 & 73.9 & 30.8 \\
\hline
TC & & 12.9 & 10.5 & & 16.0 & 16.5 \\
\hline
\end{tabular}
\caption{Yearly returns and excess returns (with respect to the index tracking portfolio shown in Table~\ref{tab 9}) in percentage, standard deviations of yearly returns (Std), Sharpe ratios, and the wealth and the cumulative transaction costs (TC) in thousands at the end of the investment period of the equally-weighted portfolio (EW) under different renewing frequencies and transaction cost rates $\tc$ with $d=100$ and daily trading frequency. The subscript $x$ corresponds to $\tc=x\%$ and the superscripts $W$ and $Q$ indicate weekly and quarterly renewing frequencies, respectively.}
\label{tab 11}
\end{table}

The cumulative transaction costs of the equally-weighted portfolio of Table~\ref{tab 11} are shown in Figure~\ref{fg EW_5}. Earlier on, the cumulative transaction costs are higher when weekly renewed than when monthly or quarterly renewed due to the large transaction costs associated with the renewal days. However, later on, the cumulative transaction costs of the weekly renewed portfolio are smaller. The reason is that the weekly renewed portfolio performs worse than the monthly or the quarterly renewed portfolio, hence the transaction costs imposed as a proportion of the portfolio wealth are also smaller.

\begin{figure}
\includegraphics[width=\textwidth]{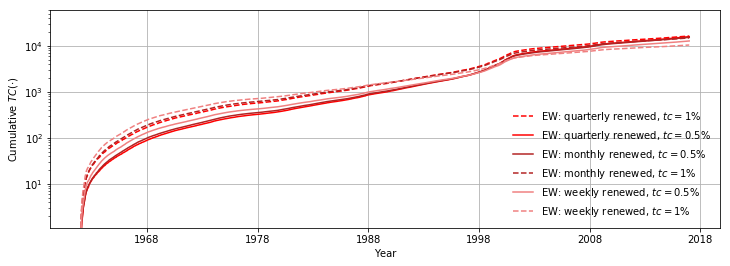}
\caption{Cumulative transaction costs on logarithmic scale of the equally-weighted portfolio (EW) under different renewing frequencies and transaction cost rates $\tc$ with $d=100$ and daily trading frequency.}
\label{fg EW_5}
\end{figure}

\subsubsection*{Varying the market size $d$}

With daily trading and monthly renewing frequencies, Figure~\ref{fg EW_6} plots and Table~\ref{tab 7} summarises the wealth processes of the equally-weighted and the corresponding index tracking portfolio under different constituent list sizes $d$. The more stocks the constituent list contains, the better the portfolio performs under the same transaction cost rate $\tc$. Again, its performance is reduced by transaction costs. Even with $d=500$ and $\tc=0.5\%$, the equally-weighted portfolio performs worse than the corresponding index tracking portfolio. In addition, the portfolio with a larger constituent list size $d$ is not necessarily more sensitive to transaction costs. Figure~\ref{fg EW_7} plots the cumulative transaction costs generated by the portfolio of Table~\ref{tab 7}.

\begin{figure}
\includegraphics[width=\textwidth]{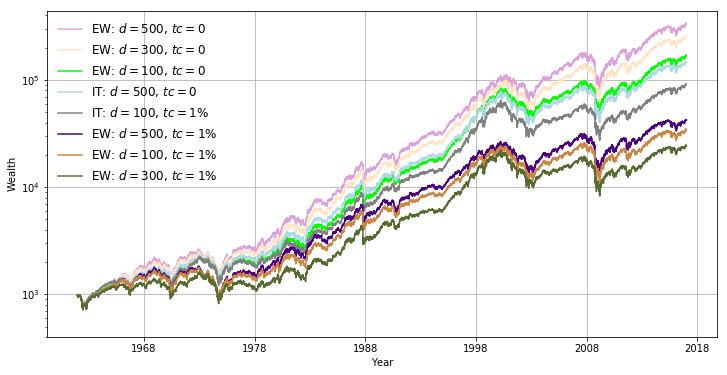}
\caption{The wealth processes of the equally-weighted portfolio (EW) and the corresponding index tracking portfolio (IT) on logarithmic scale under different constituent list sizes $d$ and transaction cost rates $\tc$ with daily trading and monthly renewing frequencies. For the index tracking portfolio, the wealth processes of the one with $d=500$ when $\tc=0$ and the one with $d=100$ when $\tc=1\%$ are plotted. The omitted wealth processes of the index tracking portfolio lie between the plotted ones.}
\label{fg EW_6}
\end{figure}

\begin{table}
\begin{tabular}{ l | c c c | c c c }
 & EW$^{300}_{0}$ & EW$^{300}_{0.5}$ & EW$^{300}_{1}$ & EW$^{500}_{0}$ & EW$^{500}_{0.5}$ & EW$^{500}_{1}$ \\
 \hline
Yearly return & 11.92 & 9.92 & 7.96 & 12.52 & 10.46 & 8.43 \\
 \hline
Excess return & 1.31 & -0.54 & -2.35 & 1.69 & -0.25 & -2.16 \\
 \hline
Std & 16.59 & 16.43 & 16.29 & 17.07 & 16.90 & 16.74 \\
 \hline
Sharpe ratio & 0.41 & 0.29 & 0.17 & 0.43 & 0.31 & 0.19 \\
\hline
Wealth & 255.3 & 93.6 & 34.3 & 332.8 & 118.4 & 42.1 \\
\hline
TC & & 21.8 & 20.5 & & 27.3 & 24.8 \\
\hline
\end{tabular}
\caption{Yearly returns and excess returns (with respect to the index tracking portfolio shown in Table~\ref{tab 10}) in precenatge, standard deviations of yearly returns (Std), Sharpe ratios, and the wealth and the cumulative transaction costs (TC) in thousands at the end of the investment period of the equally-weighted portfolio (EW) under different constituent list sizes $d$ and transaction cost rates $\tc$ with daily trading and monthly renewing frequencies. The subscript $x$ corresponds to $\tc=x\%$ and the superscripts $300$ and $500$ indicate $d=300$ and $d=500$, respectively.}
\label{tab 7}
\end{table}

\begin{figure}
\includegraphics[width=\textwidth]{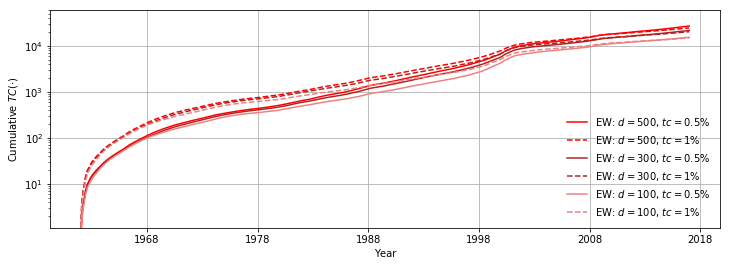}
\caption{Cumulative transaction costs on logarithmic scale of the equally-weighted portfolio (EW) under different constituent list sizes $d$ and transaction cost rates $\tc$ with daily trading frequency and monthly renewing frequency.}
\label{fg EW_7}
\end{figure}

\subsection{Entropy-weighted portfolio}\label{sub ep}

In this subsection, we consider the entropy-weighted portfolio (see Section 2.3 in \cite{MR1894767} and Example~5.3 in \cite{karatzas2017trading}), which relies on target weights
\[
\pi_{j}(\cdot)=\frac{\mu_{j}(\cdot)\log\mu_{j}(\cdot)}{\sum_{i=1}^{d}\mu_{i}(\cdot)\log\mu_{i}(\cdot)},\quad j\in\{1,\cdots,d\}.
\]

In the following, we examine the performance of the entropy-weighted portfolio under specific configurations when there are no transaction costs, i.e., when $\tc=0$, and when $\tc=0.5\%$. The performance of the entropy-weighted portfolio is less sensitive to transaction costs and is better when $\tc=0.5\%$, compared with that of the equally-weighted portfolio.

\subsubsection*{Varying the trading frequency}

As before, when backtesting the portfolio under different trading frequencies, we set the constituent list size $d=100$ and apply monthly renewing frequency. Figure~\ref{fg ETP_1} displays and Table~\ref{tab 4} summarises the wealth processes of the entropy-weighted and the corresponding index tracking portfolio under different trading frequencies. Compared with the equally-weighted portfolio summarised in Table~\ref{tab 6}, the entropy-weighted portfolio performs worse (but still outperforms the corresponding index tracking portfolio) when there are no transaction costs, i.e., when $\tc=0$. However, opposite to the equally-weighted portfolio, the weekly and the monthly traded entropy-weighted portfolio still outperforms the corresponding index tracking portfolio when $\tc=0.5\%$. 

\begin{figure}
\includegraphics[width=\textwidth]{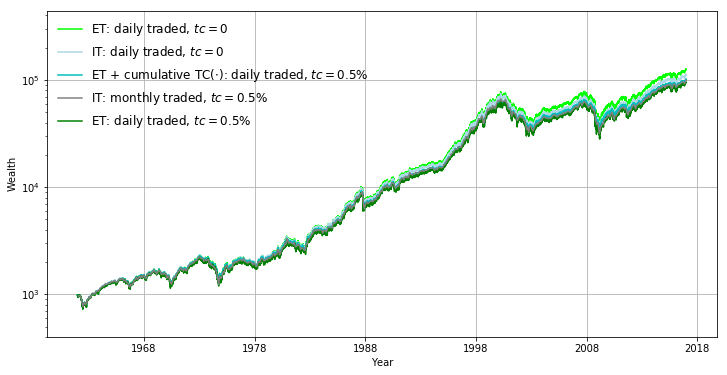}
\caption{The wealth processes of the entropy-weighted portfolio (ET) and the corresponding index tracking portfolio (IT) on logarithmic scale under different trading frequencies and transaction cost rates $\tc$ with $d=100$ and monthly renewing frequency. For both the entropy-weighted and the index tracking portfolio, the omitted wealth processes of Table~\ref{tab 4} lie between the plotted ones. The sum of the wealth process and of the cumulative transaction costs of the daily traded entropy-weighted portfolio when $\tc=0.5\%$ is also plotted. Note that the sum is below the wealth process of the daily traded entropy-weighted portfolio when $\tc=0$.}
\label{fg ETP_1}
\end{figure}

\begin{table}
\begin{tabular}{ l | c c c c | c c c c | c c c c }
 & IT$^{d}_{0}$ & ET$^{d}_{0}$ & IT$^{d}_{0.5}$ & ET$^{d}_{0.5}$ & IT$^{w}_{0}$ & ET$^{w}_{0}$ & IT$^{w}_{0.5}$ & ET$^{w}_{0.5}$ & IT$^{m}_{0}$ & ET$^{m}_{0}$ & IT$^{m}_{0.5}$ & ET$^{m}_{0.5}$ \\
 \hline
YR & 10.30 & 10.53 & 10.09 & 9.97 & 10.30 & 10.50 & 10.10 & 10.12 & 10.27 & 10.40 & 10.08 & 10.11 \\
 \hline
ER & & 0.23 & & -0.12 & & 0.21 & & 0.03 & & 0.14 & & 0.03 \\
 \hline
Std & 16.87 & 16.90 & 16.84 & 16.83 & 16.88 & 16.92 & 16.85 & 16.88 & 16.88 & 16.94 & 16.86 & 16.90 \\
 \hline
SR & 0.30 & 0.32 & 0.29 & 0.28 & 0.30 & 0.31 & 0.29 & 0.29 & 0.30 & 0.31 & 0.29 & 0.29 \\
\hline
W & 111.7 & 125.1 & 100.7 & 94.6 & 111.3 & 123.1 & 100.7 & 101.7 & 109.7 & 116.9 & 99.7 & 100.8 \\
\hline
TC & & & 2.7 & 6.4 & & & 2.5 & 4.5 & & & 2.3 & 3.4 \\
\hline
\end{tabular}
\caption{Yearly returns (YR) and excess returns (ER) in precentage, standard deviations of yearly returns (Std), Sharpe ratios (SR), and the wealth (W) and the cumulative transaction costs (TC) in thousands at the end of the investment period of the entropy-weighted portfolio (ET) and the corresponding index tracking portfolio (IT) under different trading frequencies and transaction cost rates $\tc$ with $d=100$ and monthly renewing frequency. The subscript $x$ corresponds to $\tc=x\%$ and the superscripts $d$, $w$, and $m$ indicate daily, weekly, and monthly trading frequencies, respectively.}
\label{tab 4}
\end{table}

Over a large time horizon, the loss in the portfolio wealth resulting from paying transaction costs is usually higher than the cumulative transaction costs imposed. This is exhibited in Figure~\ref{fg ETP_1}, which also plots the sum of the wealth process and of the cumulative transaction costs of the entropy-weighted portfolio when $\tc=0.5\%$. Notice that the wealth process when $\tc=0$ is above this sum. Indeed, paying transaction costs not only takes money out of the portfolio, but also deprives the opportunity for making potential gains.

\subsubsection*{Varying the renewing frequency}

With $d=100$ and daily trading frequency, we now examine the performance of the entropy-weighted portfolio applying different renewing frequencies (renewed weekly, monthly, and quarterly, respectively). Figure~\ref{fg ETP_3} displays and Table~\ref{tab 2} summarises the wealth processes of the entropy-weighted and the corresponding index tracking portfolio under different renewing frequencies. Similar to the equally-weighted portfolio, the less frequently the constituent list is renewed, the better the entropy-weighted portfolio performs. When transaction costs are imposed, its performance depends more on the renewing frequency. However, compared with the equally-weighted portfolio summarised in Table~\ref{tab 11}, the performance of the entropy-weighted portfolio is less sensitive to transaction costs under the same renewing frequency.

\begin{figure}
\includegraphics[width=\textwidth]{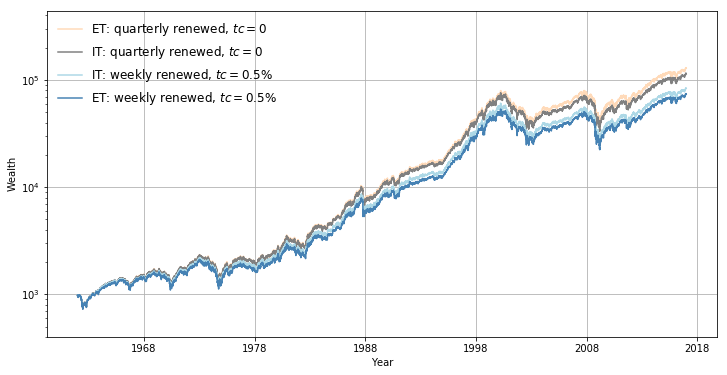}
\caption{The wealth processes of the entropy-weighted portfolio (ET) and the corresponding index tracking portfolio (IT) on logarithmic scale under different renewing frequencies and transaction cost rates $\tc$ with $d=100$ and daily trading frequency. For both the entropy-weighted and the index tracking portfolio, the omitted wealth processes of Table~\ref{tab 2} lie between the plotted ones.}
\label{fg ETP_3}
\end{figure}

\begin{table}
\begin{tabular}{ l | c c c c | c c c c }
 & IT$^{W}_{0}$ & ET$^{W}_{0}$ & IT$^{W}_{0.5}$ & ET$^{W}_{0.5}$ & IT$^{Q}_{0}$ & ET$^{Q}_{0}$ & IT$^{Q}_{0.5}$ & ET$^{Q}_{0.5}$ \\
 \hline
Yearly return & 10.14 & 10.31 & 9.73 & 9.50 & 10.34 & 10.58 & 10.20 & 10.11 \\
 \hline
Excess return & & 0.17 & & -0.23 & & 0.24 & & -0.09 \\
 \hline
Std & 16.89 & 16.93 & 16.84 & 16.84 & 16.83 & 16.86 & 16.81 & 16.81 \\
 \hline
Sharpe ratio & 0.29 & 0.30 & 0.27 & 0.26 & 0.31 & 0.32 & 0.30 & 0.29 \\
\hline
Wealth & 102.2 & 111.3 & 83.5 & 73.9 & 114.2 & 129.0 & 106.5 & 101.6 \\
\hline
TC & & & 4.2 & 7.4 & & & 2.0 & 5.8 \\
\hline
\end{tabular}
\caption{Yearly returns and excess returns in percentage, standard deviations of yearly returns (Std), Sharpe ratios, and the wealth and the cumulative transaction costs (TC) in thousands at the end of the investment period of the entropy-weighted portfolio (ET) and the corresponding index tracking portfolio (IT) under different renewing frequencies and transaction cost rates $\tc$ with $d=100$ and daily trading frequency. The subscript $x$ corresponds to $\tc=x\%$ and the superscripts $W$ and $Q$ indicate weekly and quarterly renewing frequencies, respectively.}
\label{tab 2}
\end{table}

\subsubsection*{Varying the market size $d$}

Applying daily trading and monthly renewing frequencies, we backtest the entropy-weighted portfolio under different constituent list sizes $d$ ($=100$, $300$, and $500$, respectively), as shown in Figure~\ref{fg ETP_4} and Table~\ref{tab 3}. Similar to the equally-weighted and the index tracking portfolio, the more stocks the constituent list contains, the better the entropy-weighted portfolio performs. Compared with the equally-weighted portfolio, the entropy-weighted portfolio with the same $d$ depends less on transaction costs. In particular, with $d=500$ and $\tc=0.5\%$, the entropy-weighted portfolio still outperforms the corresponding index tracking portfolio.

\begin{figure}
\includegraphics[width=\textwidth]{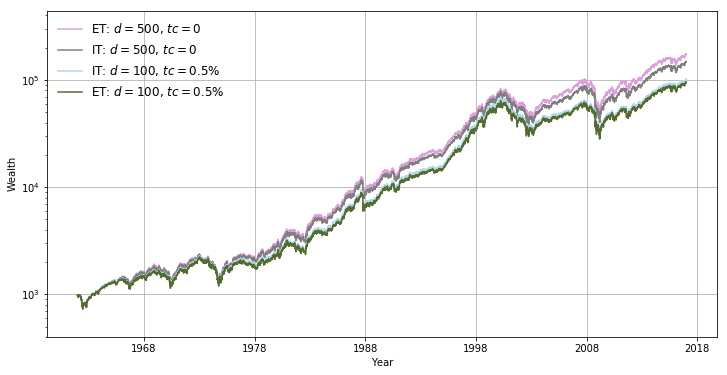}
\caption{The wealth processes of the entropy-weighted portfolio (ET) and the corresponding index tracking portfolio (IT) on logarithmic scale under different constituent list sizes $d$ and transaction cost rates $\tc$ with daily trading and monthly renewing frequencies. For both the entropy-weighted and the index tracking portfolio, the omitted wealth processes of Table~\ref{tab 3} lie between the plotted ones.}
\label{fg ETP_4}
\end{figure}

\begin{table}
\begin{tabular}{ l | c c c c | c c c c }
 & IT$^{300}_{0}$ & ET$^{300}_{0}$ & IT$^{300}_{0.5}$ & ET$^{300}_{0.5}$ & IT$^{500}_{0}$ & ET$^{500}_{0}$ & IT$^{500}_{0.5}$ & ET$^{500}_{0.5}$ \\
 \hline
Yearly return & 10.61 & 10.87 & 10.46 & 10.43 & 10.83 & 11.16 & 10.71 & 10.75 \\
 \hline
Excess return & & 0.26 & & -0.03 & & 0.33 & & 0.04 \\
 \hline
Std & 16.57 & 16.57 & 16.55 & 16.52 & 16.61 & 16.66 & 16.60 & 16.62 \\
 \hline
Sharpe ratio & 0.33 & 0.34 & 0.32 & 0.31 & 0.34 & 0.36 & 0.33 & 0.34 \\
\hline
Wealth & 132.5 & 152.5 & 123.1 & 121.2 & 147.2 & 173.2 & 139.0 & 141.0 \\
\hline
TC & & & 2.4 & 6.2 & & & 2.3 & 6.4 \\
\hline
\end{tabular}
\caption{Yearly returns and excess returns in percentage, standard deviations of yearly returns (Std), Sharpe ratios, and the wealth and the cumulative transaction costs (TC) in thousands at the end of the investment period of the entropy-weighted portfolio (ET) and the corresponding index tracking portfolio (IT) under different constituent list sizes $d$ and transaction cost rates $\tc$ with daily trading and monthly renewing frequencies. The subscript $x$ corresponds to $\tc=x\%$ and the superscripts $300$ and $500$ indicate $d=300$ and $d=500$, respectively.}
\label{tab 3}
\end{table}

\subsection{Diversity-weighted portfolio and smoothing transaction costs}\label{sec div}

One portfolio that draws much attention in Stochastic Portfolio Theory is the so-called diversity-weighted portfolio generated from the ``measure of diversity''
\[
G_{p}(x)=\left(\sum_{i=1}^{d}x_{i}^{p}\right)^{1/p},\quad x\in\left\{(y_{1},\cdots,y_{d})'\in[0,1]^{d};\sum_{i=1}^{d}y_{i}=1\right\},
\]
for some fixed $p\in(0,1)$. Without changing the relative ranking of the stocks, the function $G_{p}(\cdot)$ generates portfolio weights smaller (larger) than the corresponding market weights for stocks with large (small) market weights. This diversification property of $G_{p}$ is closely related to the implementation of relative arbitrage portfolios; see Section~7 in \cite{FK_survey} for details. Section~6.3 in \cite{MR1894767} provides a theoretical approximation of the diversity-weighted portfolio turnover. An empirical study of this portfolio using S\&P 500 market data can be found in \cite{fernholz1998diversity} and Chapter~7 of \cite{MR1894767}, as well as in Example~5 of \cite{ruf2018generalised}. 

In the following, we examine the performance of this portfolio and illustrate the tradeoff between trading with a higher frequency and paying transaction costs. To achieve this, we shall replace the market weights by a smoothed version, given by
\[
\overline{\mu}(\cdot)=\alpha\mu(\cdot)+(1-\alpha)\Lambda(\cdot)
\]
with $\alpha\in(0,1)$. Here, the moving average process $\Lambda(\cdot)=(\Lambda_{1}(\cdot),\cdots,\Lambda_{d}(\cdot))'$ is given by
\[
\begin{aligned}
\Lambda_{j}(\cdot)=
\begin{cases}
\frac{1}{\delta}\int_{0}^{\cdot}\mu_{j}(t)\mathrm{d}t+\frac{1}{\delta}\int_{\cdot-\delta}^{0}\mu_{j}(0)\mathrm{d}t\quad&\text{on~}[0,\delta)\vspace{2mm}\\
\frac{1}{\delta}\int_{\cdot-\delta}^{\cdot}\mu_{j}(t)\mathrm{d}t\quad&\text{on~}[\delta,\infty)
\end{cases},\quad j\in\{1,\cdots,d\},
\end{aligned}
\]
for a fixed constant $\delta>0$. This moving average process $\Lambda(\cdot)$ is also included in the portfolio generating function studied in \cite{schied2018model}. Then the target weights are given by
\[
\pi_{j}(\cdot)=\mu_{j}(\cdot)\left(\Xi_{j}(\cdot)-
\sum_{i=1}^{d}\mu_{i}(\cdot)\Xi_{i}(\cdot)+1\right),\quad j\in\{1,\cdots,d\},
\]
where
\[
\Xi_{j}(\cdot)=\frac{\alpha\left(\overline{\mu}_{j}(\cdot)\right)^{p-1}}{\sum_{i=1}^{d}\left(\overline{\mu}_{i}(\cdot)\right)^{p}},\quad j\in\{1,\cdots,d\}.
\]

To backtest the portfolio, we fix $d=100$, the renewing frequency to be quarterly, and the ``diversity degree'' $p=0.8$. Moreover, we compute the moving average process $\Lambda(\cdot)$ using a one-year window. To be more specific, with daily trading frequency, we set $\delta=250$; with weekly trading frequency, we set $\delta=52$. To compute $\Lambda(\cdot)$ under weekly trading frequency, we only use market weights $\mu$'s on the days when transactions are made.

\subsubsection*{Varying the convexity weight $\alpha$ and the trading frequency}

In Table~\ref{tab 5}, we summarise the wealth processes of the diversity-weighted and the corresponding index tracking portfolio under both daily and weekly trading frequencies and with three different choices for the convexity weight $\alpha$, when there are no transaction costs, i.e., when $\tc=0$, and when $\tc=0.5\%$ and $\tc=1\%$, respectively.

\begin{table}
\begin{tabular}{ l | c c c | c | c c c | c c c }
 & IT$^{w}_{0}$ & IT$^{w}_{0.5}$ & IT$^{w}_{1}$ & $\alpha$ & DW$^{d}_{0}$ & DW$^{d}_{0.5}$ & DW$^{d}_{1}$ & DW$^{w}_{0}$ & DW$^{w}_{0.5}$ & DW$^{w}_{1}$ \\
 \hline
  & & & & 0.2 & 10.36 & 10.20 & 10.03 & 10.36 & 10.20 & 10.05 \\
Yearly return & 10.34 & 10.20 & 10.06 & 0.6 & 10.43 & 10.18 & 9.93 & 10.42 & 10.23 & 10.03 \\
  & & & & 1 & 10.54 & 10.11 & 9.68 & 10.51 & 10.24 & 9.96 \\
   \hline
   & & & & 0.2 & 0.03 & 0 & -0.03 & 0.02 & 0.01 & -0.01 \\
Excess return & & & & 0.6 & 0.1 & -0.01 & -0.13 & 0.09 & 0.03 & -0.03 \\
  & & & & 1 & 0.2 & -0.09 & -0.38 & 0.18 & 0.04 & -0.09 \\
 \hline
   & & & & 0.2 & 16.84 & 16.81 & 16.79 & 16.85 & 16.83 & 16.80 \\
Std & 16.85 & 16.83 & 16.81 & 0.6 & 16.84 & 16.81 & 16.77 & 16.86 & 16.83 & 16.80 \\
  & & & & 1 & 16.84 & 16.79 & 16.74 & 16.87 & 16.83 & 16.79 \\
 \hline
   & & & & 0.2 & 0.31 & 0.30 & 0.29 & 0.31 & 0.30 & 0.29 \\
Sharpe ratio & 0.31 & 0.30 & 0.29 & 0.6 & 0.31 & 0.30 & 0.28 & 0.31 & 0.30 & 0.29 \\
  & & & & 1 & 0.32 & 0.29 & 0.27 & 0.32 & 0.30 & 0.28 \\
  \hline
  & & & & 0.2 & 115.6 & 106.2 & 97.6 & 115.1 & 106.5 & 98.5 \\
Wealth & 123.4 & 108.3 & 94.2 & 0.6 & 119.8 & 105.4 & 92.8 & 118.8 & 107.8 & 97.8 \\
  & & & & 1 & 126.2 & 101.6 & 81.7 & 124.3 & 108.3 & 94.2 \\
  \hline
  & & & & 0.2 & & 2.1 & 4.0 & & 2.0 & 3.7 \\
TC & & 3.5 & 6.3 & 0.6 & & 3.2 & 5.9 & & 2.5 & 4.6 \\
  & & & & 1 & & 5.3 & 9.0 & & 3.5 & 6.3 \\
  \hline
\end{tabular}
\caption{Yearly returns and excess returns (with respect to the index tracking portfolio (IT) summarised here and in Table~\ref{tab 8}) in percentage, standard deviations of yearly returns (Std), Sharpe ratios, and the wealth and the cumulative transaction costs (TC) in thousands at the end of the investment period of the diversity-weighted portfolio (DW) under different trading frequencies, convexity weights $\alpha$, and transaction cost rates $\tc$ with $d=100$ and quarterly renewing frequency. The subscript $x$ corresponds to $\tc=x\%$ and the superscripts $d$ and $w$ indicate daily and weekly trading frequencies, respectively.}
\label{tab 5}
\end{table}

We first consider the case when there are no transaction costs. Everything else equal, the daily traded diversity-weighted portfolio performs similarly to the weekly traded portfolio. Under either trading frequency, the smaller the convexity weight $\alpha$ is, the worse the portfolio performs. Generating the portfolio with a smaller $\alpha$ is somewhat alike to trading less frequently, as it assigns less weights on the volatile term $\mu(\cdot)$ and more weights on the stable term $\Lambda(\cdot)$ when constructing $\overline{\mu}(\cdot)$, and thus makes $\overline{\mu}(\cdot)$ less volatile.

Next, we consider the case with transaction costs. Under either daily or weekly trading frequency, a smaller convexity weight $\alpha$ tends to improve the portfolio performance when the transaction cost rate $\tc$ becomes larger. This can be useful, since decreasing $\alpha$ partially cancels out the effect of transaction costs. Moreover, when $\tc=1\%$, the daily traded portfolio with $\alpha=0.2$ performs similarly as the weekly traded portfolio with $\alpha=0.6$. This indicates that, instead of trading less frequently in order to avoid paying transaction costs, one can adjust the convexity weight $\alpha$ to reach a more favourable balance between trading frequently and paying transaction costs.

\subsubsection*{Dynamic convexity weight $\alpha$ to smooth transaction costs}

Instead of fixing $\alpha$ throughout the investment period, we could adjust $\alpha$ dynamically to speed up or slow down trading. For example, given a baseline portfolio with constant convexity weight $\alpha_{0}$, we would choose $\alpha<\alpha_{0}$ ($\alpha>\alpha_{0}$) to trade less (more) in the next period if transaction costs paid in the last period are more (less) than a certain level. In the remaining part of this example, we fix daily trading frequency and dynamically adjust $\alpha(\cdot)$. 

Let $M\geq4$ denote the total number of quarters in the investment period and let $t_{u}^{\mathrm{r}}$, for $u\in\{1,\cdots,M\}$, denote the trading days on which the constituent list is renewed. Moreover, set $t_{0}^{\mathrm{r}}=t_{0}$. On a specific renewing day $t_{u}^{\mathrm{r}}$, for $u\in\{1,\cdots,M\}$, let $\widetilde{\TC}(t_{u}^{\mathrm{r}})$ denote the averaged fictitious transaction costs relative to the wealth $V_{\alpha_{0}}(\cdot-)$ of the baseline portfolio paid in the previous period. More precisely, $\widetilde{\TC}(t_{u}^{\mathrm{r}})$ is computed as
\[
\widetilde{\TC}(t_{u}^{\mathrm{r}})=\frac{1}{\kappa_{u}}\sum_{t\in[t_{u-1}^{\mathrm{r}},t_{u}^{\mathrm{r}})}\min\left\{\frac{\TC_{\alpha_{0}}(t)}{V_{\alpha_{0}}(t-)},\xi\right\}.
\]
Here, $\kappa_{u}$ is the number of trading days within the period $[t_{u-1}^{\mathrm{r}},t_{u}^{\mathrm{r}})$, $\TC_{\alpha_{0}}(\cdot)$ is computed by \eqref{eq TC} from the baseline portfolio, and $\xi$ is a predetermined level used to make the estimate more robust. On a trading day $t$, We regard $\TC_{\alpha_{0}}(t)/V_{\alpha_{0}}(t-)>\xi$ as ``abnormal'' transaction costs relative to $V_{\alpha_{0}}(t-)$. Such large costs appear, for example, when the constituent list is changing. The level $\xi$ is determined such that the days, on which ``abnormal'' transaction costs occur, only count for a small proportion of all trading days.

\begin{figure}
\includegraphics[width=\textwidth]{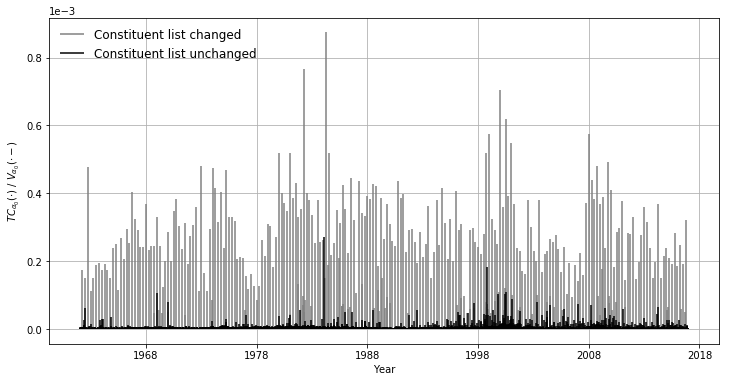}
\caption{Transaction costs $\TC_{\alpha_{0}}(\cdot)$ of the baseline portfolio relative to its wealth $V_{\alpha_{0}}(\cdot-)$, i.e., $\TC_{\alpha_{0}}(\cdot)/V_{\alpha_{0}}(\cdot-)$, paid when the constituent list is changed and unchanged, respectively, with $\alpha_{0}=0.6$, when $\tc=0.5\%$.}
\label{fg DIV_3}
\end{figure}

Figure~\ref{fg DIV_3} shows the relative transaction costs $\TC_{\alpha_{0}}(\cdot)/V_{\alpha_{0}}(\cdot-)$ when the constituent list is changed and unchanged, respectively, with $\alpha_{0}=0.6$ and $\tc=0.5\%$. Transaction costs paid when the constituent list is changed are significantly larger than when the constituent list remains the same. The days when the constituent list is changed only account for less than $5\%$ of all trading days, i.e., $M/N<0.05$, where $N$ is again the total number of trading days. 

We shall smooth the relative transaction costs $\TC(\cdot)/V_{\alpha_{0}}(\cdot-)$ by dynamically adjusting $\alpha(\cdot)$. Starting with $\alpha(t_{0})=\alpha_{0}$, the convexity weight $\alpha(\cdot)$ is piecewise constant and only updated on the renewal dates $t_{u}^{\mathrm{r}}$, for $u\in\{4,\cdots,M\}$. This reduces additional transaction costs incurred from updating $\alpha(\cdot)$. In particular, for all $u\in\{4,\cdots,M\}$, we set
\[
\alpha(t_{u}^{\mathrm{r}})=\max\left\{\min\left\{\alpha_{0}\left(1-\beta\times\overline{\TC}(t_{u}^{\mathrm{r}})\right),1\right\},0\right\}
\]
with
\[
\overline{\TC}(t_{u}^{\mathrm{r}})=\frac{\widetilde{\TC}(t_{u}^{\mathrm{r}})}{\frac{1}{4}\sum_{\nu=u-3}^{u}\widetilde{\TC}(t_{\nu}^{\mathrm{r}})}-1.
\]
Here, $\beta\geq0$ is a fixed non-negative constant that controls the sensitivity of $\alpha(\cdot)$. Hence, we compare the fictitious averaged transaction costs relative to $V_{\alpha_{0}}(\cdot-)$ within the most recent quarter to that of the past one year. The value $\overline{\TC}(\cdot)$ is positive (negative) if the baseline portfolio requires more (less) transaction costs in the most recent quarter than the last year. This will yield  $\alpha(\cdot)<\alpha_{0}$ ($\alpha(\cdot)>\alpha_{0}$) and slow down (speed up) the trading within the next quarter.

Using a baseline portfolio with constant convexity weight $\alpha_{0}=0.6$ and assuming $\tc=0.5\%$, we now estimate the effects of a dynamic convexity weight $\alpha(\cdot)$  empirically. Moreover, we set the relative transaction cost level $\xi=10^{-5}$, as the fictitious relative transaction costs $\TC_{\alpha_{0}}(\cdot)/V_{\alpha_{0}}(\cdot-)$ are less than this level on more than 95\% of all trading days. We examine the three cases $\beta\in\{0,0.05,0.1\}$. Note that $\beta=0$ yields $\alpha(\cdot)=\alpha_{0}$. With these choices of $\beta$, the portfolio with dynamic $\alpha(\cdot)$ performs similarly to the baseline portfolio; see column $V^{d}_{\tc\_0.5}$ in Table~\ref{tab 5} with $\alpha=0.6$.

The convexity weight process $\alpha(\cdot)$ corresponding to the sensitivity parameter $\beta$ is shown in Figure~\ref{fg DIV_5}. As expected, $\alpha(\cdot)$ fluctuates more rapidly with a larger $\beta$. As mentioned before, increasing $\alpha$ speeds up trading and leads to more transaction costs, while decreasing $\alpha$ has the opposite effect. Choosing $\beta$ very large results in a portfolio far away from the baseline portfolio. This dependence on $\beta$ is illustrated in Figure~\ref{fg DIV_7}, which plots the total quadratic variation $\mathrm{QV}$ of relative transaction costs $\TC(\cdot)/V_{\alpha_{0}}(\cdot-)$, computed as
\[
\mathrm{QV}=\sum_{l=1}^{N}\left(\frac{\TC(t_{l})}{V_{\alpha_{0}}(t_{l}-)}-\frac{\TC(t_{l-1})}{V_{\alpha_{0}}(t_{l-1}-)}\right)^{2},
\]
for different sensitivity parameters $\beta$. The total quadratic variation $\mathrm{QV}$ is a measure of volatility. Figure~\ref{fg DIV_7} suggests that choosing $\beta\approx0.05$ minimises $\mathrm{QV}$.

\begin{figure}
\includegraphics[width=\textwidth]{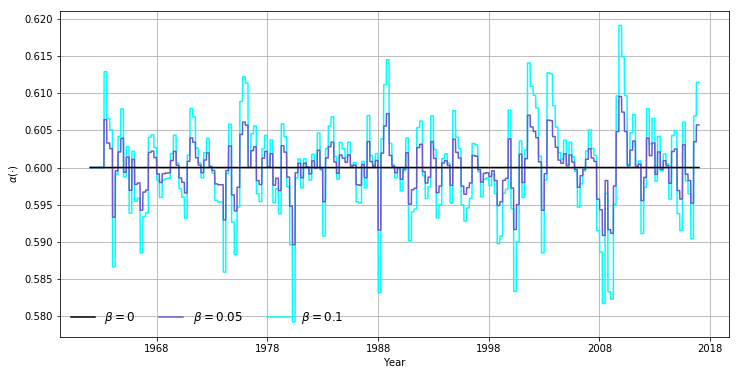}
\caption{Convexity weight process $\alpha(\cdot)$ for different sensitivity parameters $\beta$ with $\alpha_{0}=0.6$, when $\tc=0.5\%$.}
\label{fg DIV_5}
\end{figure}

\begin{figure}
\includegraphics[width=\textwidth]{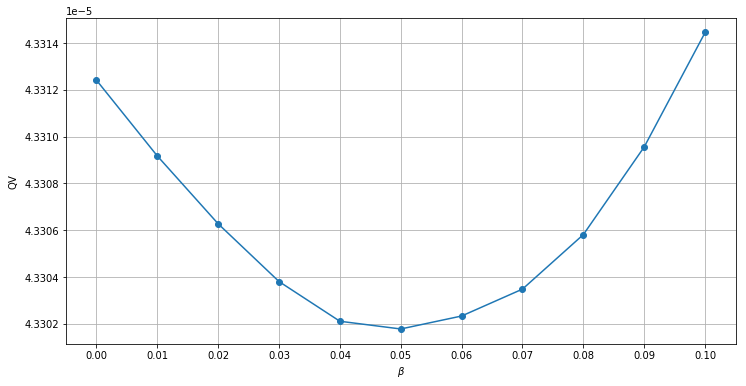}
\caption{Quadratic variation $\mathrm{QV}$ of relative transaction costs $\TC(\cdot)/V_{\alpha_{0}}(\cdot-)$ for different sensitivity parameters $\beta$ with $\alpha_{0}=0.6$, when $\tc=0.5\%$.}
\label{fg DIV_7}
\end{figure}

\section{Conclusion}\label{sec C}

In this paper, we empirically study the impact of proportional transaction costs on systemically generated portfolios. Given a target portfolio, we provide a scheme to backtest the portfolio using total market capitalization and daily stock return time series. Implementing this scheme, we examine the performance of several portfolios (the index tracking portfolio, the equally-weighted portfolio, the entropy-weighted portfolio, and the diversity-weighted portfolio), assuming various transaction cost rates, trading frequencies, portfolio constituent list sizes, and renewing frequencies.

As expected, everything else equal, a portfolio performs worse as transaction costs are higher and the portfolio renewing frequency of the underlying constituent list is higher. In the absence of transaction costs, trading under a higher frequency leads to better portfolio performance. However, in the presence of transaction costs, implementing a higher trading frequency can also result in larger transaction costs and reduce the portfolio performance significantly. Hence, trading under an appropriate frequency is necessary in practice. In addition, with or without transaction costs, a more diversified portfolio containing more stocks usually performs better.

The empirical results indicate that the equally-weighted portfolio performs well relative to the index tracking portfolio when there are no transaction costs. However, the performance of the equally-weighted portfolio is very sensitive to transaction costs. Although the entropy-weighted portfolio performs a bit worse than the equally-weighted portfolio (but still outperforms the index tracking portfolio) when there are no transaction costs, its performance depends much less on transaction costs, compared to the equally-weighted portfolio.

Last but not the least, we propose a method to smooth transaction costs. Without changing the trading frequency, this method is similar to altering the trading speed dynamically.

\bibliography{TC}
\bibliographystyle{chicago}
\end{document}